\documentclass[preprint,11pt]{elsarticle}
\usepackage[suffix=]{epstopdf}

\usepackage{geometry}
\usepackage{color}
\usepackage{amssymb,amsmath}

\usepackage{tikz}

\usepackage{chngcntr}

\usepackage{amsfonts}
\usepackage{graphics}

\def\N{\mathbb N}

\newtheorem{theorem}{Theorem}
\newtheorem{remark}{Remark}

\newenvironment{proof}{\paragraph{Proof}}{\hfill$\square$}

\newtheorem{prop}{Proposition}
\begin{document}

\title{Rich dynamics and anticontrol of extinction in a prey-predator system }
\vspace{5mm}

\author [rm1,rm2]{Marius-F. Danca}\corref{cor1}
\author [rm3,rm4]{Michal Fe\v{c}kan}
\author [rm5,rm6]{Nikolay Kuznetsov}
\author [rm7] {Guanrong Chen}
\cortext[cor1]{Corresponding author}
\address[rm1]{Dept. of Mathematics and Computer Science, Avram Iancu University of Cluj-Napoca, Romania}
\address[rm2]{Romanian Institute of Science and Tecçhnology, Cluj-Napoca, Romania\\
email: danca@rist.ro, orcid.org/0000-0001-7699-8709}
\address[rm3]{Dept. of Mathematical Analysis and Numerical Mathematics, Faculty of Mathematics, Physics and Informatics, Comenius University in Bratislava, Slovak Republic,}
\address[rm4]{Mathematical Institute, Slovak Academy of Sciences, Slovak Republic\\
email:Michal.Feckan@fmph.uniba.sk }
\address[rm5]{Dept. of Applied Cybernetics, Saint-Petersburg State University, Russia}
\address[rm6]{Dept. of Mathematical Information Technology, University of Jyv\"{a}skyl\"{a}, Finland\\
email: nkuznetsov239@gmail.com}
\address[rm7]{Dept. of Electronic Engineering, City University of Hong Kong, Hong Kong, China\\
email: eegchen@cityu.edu.hk}

\begin{abstract}

This paper reveals some new and rich dynamics of a two-dimensional prey-predator system and to anticontrol the extinction of one of the species. For a particular value of the bifurcation parameter, one of the system variable dynamics is going to extinct, while another remains chaotic. To prevent the extinction, a simple anticontrol algorithm is applied so that the system orbits can escape from the vanishing trap. As the bifurcation parameter increases, the system presents quasiperiodic, stable, chaotic and also hyperchaotic orbits. Some of the chaotic attractors are Kaplan-Yorke type, in the sense that the sum of its Lyapunov exponents is positive. Also, atypically for undriven discrete systems, it is numerically found that, for some small parameter ranges, the system seemingly presents strange nonchaotic attractors. It is shown both analytically and by numerical simulations that the original system and the anticontrolled system undergo several Neimark-Sacker bifurcations. Beside the classical numerical tools for analyzing chaotic systems, such as phase portraits, time series and power spectral density, the '0-1' test is used to differentiate regular attractors from chaotic attractors.
\end{abstract}


\maketitle

\section{Introduction}

During the last few decades, prey-predator systems have received a renewal of attention (see e.g. \cite{prey2,prey4,muica,prey6,baga_altu,prey7,lafel,dancax,prey8,prey9,prey1}).
This paper considers the discrete variant of a continuous-time Lotka-Volterra prey-predator system \cite{dancax,lafel}, in which the competition between two species $x$ and $y$ is modeled by the following iterative equations:

\begin{equation}\label{eq1}
\begin{split}
&  x_{n+1}=ax_n(1-x_n)-bx_ny_n,\\
 & y_{n+1}=dx_ny_n,
 \end{split}
  \end{equation}
where $a,b,d$ are positive parameters and $x_n$ and $y_n$ denote the prey and the predator densities respectively in year (generation) $n$, $bx$ represents the number of prey individuals consumed per unit area per unit time by an individual predator, and $dx_ny_n$ is the predator response \cite{prey2}.

Some of the complex dynamical behaviors and stability aspects of system \eqref{eq1} are presented in \cite{dancax}. In this paper, more rich dynamics are revealed, such as hyperchaotic attractors, chaotic attractors in the sense of Kaplan-Yorke, and strange nonchaotic attractors (SNAs). Also, an anticontrol algorithm is designed and utilized for preventing the extinction dynamics of the predator variable $y_n$.

The paper is structured as follows: Section 2 presents the system dynamics, including Neimark-Sacker bifurcation, quasiperiodic attractors, stable attractors, chaotic attractors and hyperchaotic attractors. Section 3 presents a particular chaotic attractor at $a=4$ when the variable $y_n$ vanishes after a few thousands of iterations. Also, an anticontrol algorithm is designed and used to prevent $y_n$ from extinction. Some comments are given in the Conclusion section, where SNAs are briefly discussed.

\section{System dynamics}
From the biology standpoint, assume that the system is defined on the first closed quadrant $\mathbb{R}_+^2=\{(x,y):x\geq0,y\geq0\}$, with initial conditions satisfy $x_0:=x(0)>0$ and $y_0:=y(0)>0$.

A numerical approximation of the boundedness domain $\mathcal{D}\subset \mathbb{R}_+ ^2$ in the parameter space $(a,d)$, on which the system is defined,
is colored yellow in the lattice $[0,8]\times[0,10]\subset \mathbb{R}^2_+$ in Fig. \ref{fig1} (a). Outside the yellow domain, i.e. for $a\notin[0,8]$ and $d\notin[0,10]$, the system might be divergent (grey region). As can be seen, the frontier of the boundedness domain presents some zones whose structures are rater complex, with fractal characteristics (zoomed zone $D1$) and rectilinear zones (zoomed zone  $D2$).

Fig. \ref{fig00} reveals the extremely rich and complex system dynamics: regular dynamics, such as mode-locking (or stable periodic orbits), quasiperiodicity (or invariant circles), and chaotic and hyperchaotic dynamics. Generally, relatively small periodic windows are immersed in larger quasiperiodic windows.

The parameter $b$ does not influence the system dynamics. Therefore, hereafter in all numerical experiments, set $b=0.2$ and, unless specified, $d=3.5$.

Local (finite-time) Lyapunov exponents (LEs), are determined numerically from the system equations. Except for two different attraction basins, which appear for some range of the bifurcation parameter $a$ within the existence domain, the values of the local LEs are approximatively the same. Hereafter, the local LEs are simply called LEs and the spectrum is denoted $\Lambda=\{\lambda_1,\lambda_2\}$, with $\lambda_1>\lambda_2$.

Unless specified, the iteration number, necessarily to obtain meaningful numerical results, is set to $n=5e5$.
Denote by $P_i$, $i=1,2,...$, some most important points in the partition of the parameter range $a\in[2,4]$ and by $NS$ the point corresponding to the Neimark-Sacker bifurcation. Zero LEs are considered having at least 4 zero decimals, i.e. with error less than $1e-5$.

In order to analyze the qualitative changes of the system and to follow the changes in the system dynamics, consider the bifurcation diagram on the plane $(a,x)$, together with Lyapunov spectrum $\Lambda$, for $a\in[2,4]$, where the most important dynamics are shown in Fig. \ref{fig00}.
Note that the particular shape of the maximal LE begining from $P_1$ and $P_2$, resembles the existence of SNAs, a notion described first by Grebogi et al. in 1984 \cite{sna}.
Supplementarily, the binary '0-1' test (see Apendix  \ref{ee}) is utilized to distinguish clearly regular attractors from chaotic attractors. This test indicates a value close to 0 for regular dynamics, and a value that tends to 1 for chaos.

To study the nature of the attractors, phase plots, time series, LEs, normalized power spectral density (PSD) and '0-1' test are utilized. As the time series are real, PSD is two-sided symmetric and, therefore, only the left-side is discussed here.

The PSD is used to unveil the birth of new frequencies, one of the presumably paths to chaos for this system. From the evolution of the asymptotic growth rate $K$ as function of $a$ (Fig. \ref{fig00} (c)), one can see that, except the range $a\in(3.2,3.25)$, where the transition from regular motion to chaotic motion is completed gradually. This fact suggests the existence of SNAs (see Conclusion and Discussions Section). For the other intervals of $a$, $K$ changes abruptly between values 0 and 1.

Attractors with: a) $\lambda_{1,2}<0$, a single frequency in PSD (with potential harmonics) and $K=0$, are considered to be \emph{stable periodic orbits}; b) $\lambda_1>0$, $\lambda_2<0$, a broadband in PSD and $K=1$, are considered to be \emph{chaotic attractors}; c) $\lambda_1=0$, $\lambda_2<0$, several discrete peaks in PSD and $K=0$, whose orbit points in the phase space never repeat itself, are called \emph{quasiperiodic orbits} (or \emph{invariant circles}); d) $\lambda_{1,2}>0$, broadband in PSD and $K\approx 1$, are considered to be \emph{hyperchaotic attractors}; e) $\lambda_1>0>\lambda_2$, for which $\sum\lambda_i>0$, are called \emph{Kaplan-Yorke (K--Y) chaotic attractors} \cite{kaplan}, and f) $\lambda_2\leq \lambda_1\leq 0$ and $0<K<1$, are called \emph{strange nonchaotic attractors}. Stable periodic orbits and quasiperiodic orbits are \emph{regular orbits}.

The system reveals a lot of interesting qualitative changes in its dynamics. In this paper the interest is mainly on the Neimark-Sacker ($N-S$) bifurcation, one of the main phenomenon of this system, and on its chaotic dynamics.

An important characteristic of the system is the multistability, or coexistence of attractors (see the zoomed portion in the bifurcation diagram Fig. \ref{fig00} (d), and also Fig. \ref{figg} (a)). Thus, for the same value of the parameter $a$, i.e., $a=3.381$ (with $d=3.5$) and different initial conditions $(0.5,1.2)$ and $(0.006,3.65)$, one obtains two different attractors: a quasiperiodic orbit (red plot) and a stable periodic orbit (blue plot), respectively (Fig. \ref{figg} (b)). Also, for $a=3.3815$ and initial conditions $(0.5,1.2)$ and $(0.006,3.65)$, other two different attractors emerge (Fig. \ref{figg} (c)). For clarity, the attractive points of the stable orbit are represented by filled circles. Comparing Figs. \ref{figg} (b), (c), one can see that the quasiperiodic orbit (red plot), composed of 6 incomplete island-like sets, degenerates for $a$ passing through a crisis value $a^*$ (slightly lower than $a=3.3815$), where the attractor begins to break. Thus, from the triangular shape, they transform to 6 sets of triplets (points), marking the corner of the former island-like set scenario and forming the attractive periodic points of a period-18 stable orbit. This trifurcation-like scenario resembles the birth of chaos via the Period Three Theorem. Contrarily, the 12 blue points of the stable orbit transform to 12 closed quasiperiodic island-like sets.

Regarding the coexistence of different attraction basins, determined by the coexisting attractors Fig. \ref{fig00} is obtained with the initial condition $(0.006, 3.65)$. Unspecified initial conditions mean that they do not influence the numerical results.

Denote the three fixed points of the system by
\begin{equation*}
X_1^*(0,0), X_2^*\Big(1-\frac{1}{a},0\Big), X_3^*\Big({\frac{1}{d},\frac{a}{b}\Big (1-\frac{1}{d}\Big)-\frac{1}{b}\Big)}\Big),
\end{equation*}
Consider, next, $a$ and $d$ as bifurcation parameters in the plane $(a,d)$.

\begin{prop} \label{prop1}For all $b\in \mathbb{R}_+$:

(i) If $a<1$ $X_1^*$ is stable for all $d\in \mathbb{R}_+$;

(ii) If $d-a/(a-1)<0$, with $a\in(1,3)$, $X_2^*$  is stable for all $d\in \mathbb{R}_+$;

(iii) If $\max\{3a/(a+3), a/(a-1)\}<d<2a/(a-1)$, $a>1$, $X_3^*$ is stable.

\end{prop}
\begin{proof} See the proof in Appendix \ref{aa}.
\end{proof}

The stability domains are presented in Fig. \ref{fig1} (b) over the boundedness domain (Fig. \ref{fig1} (a)). Color red represents the stability domain of the fixed point $X_1^*$, blue for the fixed point $X_2^*$ and green for $X_3^*$. On the remaining domain, fixed points are unstable.

\subsection{Quasiperiodic regime}
Because the quasiperiodic windows are prevalent, the $N-S$ bifurcation, which generates quasiperiodicity, is now studied both analytically and numerically.
Denote its complex eigenvalues by $e_{1,2}^c$ (see Appendix \ref{aa}). The $N-S$ bifurcation occurs when the fixed point $X_3^*$ changes stability, for which the following conditions are satisfied \cite{guk}\footnote{The formula for the first Lyapunov coefficient (genericity condition), which indicates the type of $N-S$ bifurcation, is not considered here.}:

(i) $|e_{1,2}^c(a,d)|=1$;

(ii) $\frac{\partial}{\partial d}(|e_{1,2}^c(a,d)|)=k\neq0$ (transversality condition);

(iii) $(e_{1,2}^c(a,d))^j\neq1$ for $j=1,2,3,4$.


The stability of $X_3^*$ changes as $a$ and $d$ are varied. Therefore, the following theorem can be established
\begin{theorem}\label{ho1}
For
\begin{equation}\label{ns}
a=\frac{d}{d-2},
\end{equation}
with $d>2$, the system \eqref{eq1} undergoes a $N-S$ bifurcation.

\end{theorem}

\begin{proof} See Appendix \ref{bb}.
\end{proof}

In the parametric plane $(a,d)$, the relation \eqref{ns} represents the $N-S$ curve, denoted as $AB$ in Fig. \ref{fig1} (c) (red plot). Thus, $N-S$ bifurcations appear at points $NS(a,d)$ situated along the curve $AB$, for $a\in(1.41,4)$ and $d\in(2.67,6.88)$.

Hereafter, the parameter $d$ is set to $d=3.5$ (see the dotted line in Figs. \ref{fig1} (b) and (c)). According to the relation \eqref{ns}, the bifurcation occurs at the point $NS$ in Fig. \ref{fig00}, with $a=a_{NS}=2.3333...$, in perfect agreement with the numerical results.

The $N-S$ bifurcation is supercritical because, at the fixed point $X_3^*$, it loses its stability, giving birth to a closed invariant curve (topologically equivalent to a circle).
Once $a$ crosses the point $NS$ along the line $d=3.5$, the system passes from the stable fixed point $X_3^*$ to an unstable fixed point, and an one-frequency invariant closed curve is born, whose size in the phase plane grows as the parameter $a$ increases. To understand better this phenomenon, consider Fig. \ref{sapte}. Fig. \ref{sapte} (a) represents the phase portraits, while Fig. \ref{sapte} (b), the PSD. In Figs. \ref{sapte} (i), the case of $a=2.33<a_{NS}$ (before bifurcation) is considered: the phase portrait shows that the orbits are attracted by $X_3^*$. The PSD indicates that there are no frequencies yet. In Figs. \ref{sapte} (ii), the case of a slightly bigger value of $a$, $a=2.3334>a_{NS}$, is considered. The three-dimensional simulative result in the space $(a,x,y)$ shows that $X_3^*$ loses its stability and the orbits from outside (and inside) are the born invariant circle $\Gamma$, which are attracted by the quasiperiodic orbit. PSD reveals the birth of the first fundamental frequency (first harmonic)\footnote{It is known that a physical mark of the $N-S$ bifurcation is the apparition of a new additional frequency in the time series.} $f_0\approx0.1339$ and the place of the first harmonic $f_1$.
Fig. \ref{sapte} (iii) presents the case of $a=2.3545$, at a relatively larger distance from $a_{NS}$. For this value of $a$, the size of $\Gamma$ still increases and other two harmonics (modulation frequencies) of $f_0\approx0.1343$ are born, being distributed in the PSD as follows: $f_k=f_0+k\delta_1=(k+1)f_0$, $k=1,2$, equidistant at a constant offset of $\delta_1\approx0.1343$ to each other. Also, the second main frequency $f_{21}\approx0.4628$ appears, to the right of $f_2$, at the distance of $\delta_2\approx0.0599$.
The quasiperiodic oscillations are non-resonant, $f_0/f_{21}$, which is an irrational number.


If one considers the complex eigenvalues $e_{1,2}^c(a,d)$ of $X_3^*$ for $d>a/(a-1)$ in the exponential form: $e_1^c(a,b)=A(a,d)e^{i\alpha}$ (see Appendix \ref{aa}), and $e_2^c(a,b)=\overline{e_1^c}=A(a,d)e^{-i\alpha}$, then the argument $\alpha$ can be obtained via the relation

\[
\alpha=\arctan\frac{Im(e_{1}^c)}{Re(e_{1}^c)}=\arctan \bigg(\frac{1}{2}\frac{\sqrt{4(a-1)d^2-4ad-a^2}}{2d-a}\bigg).
\]
Calculations show that a simpler form of $\alpha$, determined at the bifurcation points $d^*=2a/(a-1)$ and situated on the curve $N-S$, can be expressed as
\begin{equation}\label{alfa}
\alpha(d^*)=\arccos\frac{5-a}{4}.
\end{equation}
\indent i) If $\alpha(d^*)/2\pi$ is a rational number, i.e.,

\begin{equation}\label{mode}
\alpha(d^*)/2\pi=m/n,
\end{equation}
with $m,n$ being some positive relative-prime integers, then one has a \emph{periodic regime}, \emph{mode-locked state}, or \emph{stable periodic orbit}, and the iterated points repeat (after relatively long transients). In this case $n$ is the period of the orbit, while $m$ represents its multiplicity.

ii) If $\alpha(d^*)/2\pi$ is an irrational number, the orbit points tend to fill an invariant (dense) closed orbit, which never repeats itself. The orbit is called \emph{quasiperiodic}, or an \emph{invariant circle}, though the orbit may never exactly repeat itself, and the motion remains \emph{regular}.

Because the sudden appearance or disappearances of certain dynamics in the system \eqref{eq1} (such as transitions between quasiperiodicity, mode-locking, chaos and hyperchaos) as the parameter $a$ is varied, this resembles the crisis in strange chaotic attractors (as introduced by Grebogi et al. \cite{cris}). Similarly, this phenomenon will be called \emph{crisis}.

At point $P_1(2.718)$, after the $N-S$ bifurcation, the quasiperiodicity continues till a major periodic window appears, which starts at $a=a_{P_1}$, and the quasiperiodic attractor is suddenly destroyed, allowing the formation of a stable window. In Fig. \ref{opt}, few cases around $a=a_{P_1}$ are presented. Fig. \ref{opt} (a) represents the bifurcation diagram embedding the periodic widow and Fig. \ref{opt} (b) presents LEs which, within the periodic window, are both negative. A boundary crisis opens a window, when the quasiperiodic attractor is suddenly destroyed in favor of the stable window, and a reverse crisis appears at the end of the stable window, when the stable period-7 orbit suddenly disappears making room for the new quasiperiodic attractor.
 To underline the differences in the system dynamics around this point, Figs. \ref{opt} (i), (ii) present two representative orbits: $a=2.6$ (quasiperiodic orbit) situated at about the middle of the interval $(a_H,a_{P_1})$, and $a=2.72$, slightly after $P_1$ (stable periodic orbit), respectively. For each case, Figs. \ref{opt} (c)-(e) present the phase portraits, partial time series, and PSD. The quasi-perioic orbit in Fig. \ref{opt} c(i) contains 7 points (all unstable at this value of $a$), plotted with filled circles.
  For the quasiperiodic orbit, the time series in Fig. \ref{opt} d(i) reveals that the orbit never repeats.
Fig. \ref{opt} e(i) shows the two main frequencies, $f_0\approx0.1398$ and $f_{21}\approx0.4409$, and their harmonics (offsets are $\delta_1\approx0.1398$ and $\delta_2\approx0.0215$).
The two main frequencies (and their harmonics) are incommensurate, so the oscillations are quasiperiodic and non-resonant. Fig. \ref{opt} c(ii) plotts the stable period-7 orbit. In this case, points position is slightly different from the case of $a=2.6$, because of the slight increment of the parameter $a$. Now, only the frequencies $f_{0,1,2}$ are clearly visible. The second frequencies $f_{k1}$, $k=0,1,2$, collide with the first frequency and its harmonics leading to the disappearance of the quasiperiodic orbit. The obtained oscillations are resonant (mode-locking).

Point $P_2(3.1)$ presents similar characteristics with $P_1$ and, therefore, is omitted.

For $a\in(a_{NS},a_{P_3})$, with $a_{P_3}=3.194$, the quasiperiodicity undergoes crisis at points $P_1$ and $P_3$. Thus, at these points, the maximal LE, $\lambda_1$, leaves suddenly the zero value, becoming negative, and opens periodic windows for small parameter ranges.

Note that, for $a\in[2,4]$ and $d=3.5$, only two independent frequencies have been found. Also, the invariant orbits suggest that the discrete predator-prey system follows dynamical
behaviors that are homogeneous in space and quasiperiodically oscillating in time \cite{muica}.

\subsection{Finding stable periodic orbits and the periods}\label{maimulte}
{i) \emph{Stable cycle}.

1. Suppose one intends to obtain a stable periodic orbit, e.g. of period 7. From the relations \eqref{alfa}-\eqref{mode}, with $m=1$ and $n=7$, one obtains $a=2.507040792$ from $\arccos~((5-a)/4)-2\pi/7)=0$ and $d=3.327985277$ from \eqref{ns}. This new $N-S$ bifurcation point, $NS_1$, can be viewed in Fig. \ref{fig1} (c). Figs. \ref{cica_patru} (i) present the phase portraits of two values of $(a,d)$ slightly near $(a,d)_{NS_1}$.

2. To obtain a stable orbit of period 6 ($m=1$ and $n=6$) from the equation $\arccos~((5-a)/4)-2\pi/6)=0$, one has $a=3$ and the corresponding $d=3$ (point $NS_2$ in Fig. \ref{fig1} (c)). Figs. \ref{cica_patru} (ii) show two values of $(a,d)$, slightly near $(a,d)_{NS_2}$.

Note that, at all $N-S$ bifurcation points, both LEs are zero.

ii) \emph{Period}. To find the period of a stable orbit (which implicitly has a single main frequency) with a relatively small error (order of $1e-3$), the PSD can be used as follows where the period of a stable orbit, $T$, is defined as $T=1/f$, which is the fundamental frequency.

1. For example, to find the period of the stable orbit for $a=2.72$ (see Figs. \ref{opt} (ii)), from the PSD one obtains $f_0\approx0.1429$ and, therefore, $T=1/f_0\approx6.9979\approx 7$, as obtained numerically. Similarly, the harmonics $f_{1,2}$ indicate the repetitions after $14$ steps and $21$ steps.

2. For $a=a_{P_5}=3.36$ (see Fig. \ref{zece} (a)), $f_0\approx0.1667$ and $T=1/f_0\approx5.9988\approx6$.

3. Consider $a=3.575$ (Figs. \ref{unspe} (i)). Denoting the first frequency by $f_\text{I}\approx0.0714$, one can see the previous second frequency and its harmonics moved in the frequency space which, together with the first frequency and its harmonics, form a new single frequency series peeks, $kf_\text{I}$, $k=1,2,3...,7$, with offset of $\delta_3\approx0.0357$ (Fig. \ref{unspe} c(i)). Therefore, $T=1/f_\text{I}\approx13.9997\approx14$. In this case, the second main frequency and its harmonics suggest a bifurcation, when the second main frequency (and its harmonics) are born as half of the main frequency (and its harmonics).

\subsection{Chaotic regime}
Point $P_3(3.194)$ represents the center of an extremely narrow window, which enhances the first chaotic behavior. The window begins with a boundary crisis, where the quasiperiodic oscillations suddenly disappear and a stable periodic orbit with an extremely large period (see the periodic window in Fig. \ref{noua} (a), yellow plot), and ends with a reverse crisis, where the stable periodic orbit vanishes and chaotic oscillations appear. Fig. \ref{noua} (b) presents the phase portrait of a quasiperiodic orbit for $a=3.193<a_{P_3}$. The zero maximal LE and the fact that the orbit is continuously filled, never repeats, indicates that the orbit is quasiperiodic (see also the Lyapunov spectrum $\Lambda$ in Fig. \ref{noua} (c)). For a slightly larger value of $a=3.1945>a_{P_3}$, the system behaves chaotically (see Fig. \ref{noua} (d)). The transient (dotted plot) of the quasiperiodic attractor for $a=3.1999$ (Fig. \ref{noua} (e)) indicates the transformation (crisis) of the previous chaotic attractor in Fig. \ref{noua} (d), and then the chaotic orbit is broken into 45 quasiperiodic island-like sets. The period-45 small amplitude invariant circles, which compose the attractor, surrounding unstable points. In Fig. \ref{noua} (f), another chaotic attractor, in a different shape, appears for $a=3.3<a_{P_4}$. Since $\sum \lambda_i=0.065-0.018>0$, this attractor is K--Y chaotic. Note that, for $a=3.24533$, the system evolves along a stable multi-period orbit (similar to the orbit in Fig. \ref{zece} (a), but with a higher period). Right after this value, $a=3.3$ (Fig. \ref{noua} (f)), the system has chaotic behavior.

In the window located between the points $P_4(3.325)$ and $P_9(3.425)$, beside the above-mentioned bistability, the system presents a plethora of attractors (Figs. \ref{zece}): stable periodic orbits (point $P_5(3.36)$, Fig. \ref{zece} (a)), quasiperiodic orbits (point $P_6(3.368)$, Fig. \ref{zece} (b)), chaotic orbits (point $P_7(3.39)$, Fig. \ref{zece} (c)) and hyperchaotic orbits (point $P_8(3.394)$, Fig. \ref{zece} (d)). After a chaos doubling process, the previous chaotic attractor shown in Fig. \ref{zece} (c) gives birth to 12 small chaotic attractors composing the hyperchaotic attractor. The attractor is considered ``weak'' because $\lambda_2$ is slightly larger than 0 ($\lambda_2\approx3e-3$). This window (between $P_4(3.325)$ and $P_9(3.425)$) begins with a stable period-6 orbit, which was born after a sudden extinction of the chaotic orbit at $P_4$, indicating a boundary crisis, and it ends with the chaotic orbit at $P_9(3.425)$. The transients, shown for clarity only in Fig. \ref{zece} (a), are chaotic, resembling the shape of the chaotic orbit before $P_4$ (see the chaotic orbit for $a=3.3<a_{P_4}$ in Fig. \ref{noua} (e)).

The parameter interval $(P_9(3.425),P_{10}(3.57))$ corresponds to a stable period-7 window (see Fig. \ref{fig00}).

The window $(P_{10}(3.57),P_{13}(3.58))$ is predominantly periodic beginning with a boundary crisis, which generates a stable period-7 orbit (see Fig. \ref{fig00} (e) and the detail $D$, which indicates two nearby double branches). Note that, for $a\in (a_{P_{10}},a_{P_{13}})$, $\lambda_1<0$, and at $P_{11}(3.571)$, its modulus is apparently small (of order $1e-3$) but is still negative. Therefore, there exists a stable periodic orbit with a single main frequency $f_I$ and 6 other harmonics, $kf_I=(k-1)\delta_3$, $k=2,3,...,7$, with $\delta_3\approx0.0357$ being the offset between two consecutive peeks (see Figs. \ref{unspe} (i)). Actually, the second main frequency and its harmonics still exist, but they are commensurate with the first frequencies. Between $P_{11}(3.571)$ and $P_{12}(3.578)$, for $a=3.575$, because of a period-doubling process, one obtains a stable period-14 orbit (see also Subsection \ref{maimulte} ii).
Between $P_{12}(3.578)$ and $P_{13}(3.58)$, the system presents quasiperiodic oscillations (see Figs. \ref{unspe} (ii), where the case of $a=3.579$ is considered). Because of the trifurcation of the first main frequency \cite{kapa} and of its harmonics (two new peeks at the left and right sides, Fig. \ref{unspe} c(ii), red plot), the second frequency is born. Therefore, the stability of the previous period-14 orbit is destroyed and 14 invariant small-amplitude circles surround the period-14 points, which are now unstable. The difference between the period-14 stable orbit in Figs. \ref{unspe} (i) and the quasiperiodic orbit in Figs. \ref{unspe} (ii) is clearly unveiled beside the PSD, by the partial time series (Figs. \ref{unspe} b(i), b(ii)).
Next, via a reverse boundary crisis, at $P_{13}$, the quasiperiodic orbits transform into chaotic orbits. For $a=3.581$, a K--Y chaotic attractor is shown in Fig. \ref{unspe} (iii).
The several discrete peaks of the quasiperiodic orbit in Fig. \ref{unspe} c(ii) transform now into a broadband (Fig. \ref{unspe} c(iii)), where the former frequencies $f_I$ and its harmonics still can be seen.

The window between $P_{13}(3.58)$ and $P_{14}(3.965)$ (see the detail in Fig. \ref{fig00} (e)), starts with an interior crisis, when the quasiperiodicity suddenly vanishes, making room for to hyperchaotic windows, which alternate with narrow quasiperiodic windows. The shapes of the hyperchaotic attractors are similar with the shapes of the chaotic attractors, for $a<a_{P_{13}}$ (see e.g. Fig. \ref{unspe} c(iii)). Therefore, they are not presented here. A quasiperiodic orbit for $a$ within a small neighborhood of $3.5826$ is presented in Fig. \ref{doispe}. The zoomed bifurcation diagram in Fig. \ref{doispe} (a), for $a\in[3.5824,3.5829]$, reveals the fact that within the first narrow window there exist quasiperiodic orbits (see the horizontal bifurcation branches with a relative thickness indicating the quasiperiodicity). After that, in the next window, there appear stable orbits, similarly with the case in Fig. \ref{figg} (a). Fig. \ref{doispe} (b) indicates 14 spots (numbered in an aleatory order), which reveal quasiperiodic islands-like sets.

The last window $[P_{14}(3.965), P_{17}(3.9999)]$ (see the detail in Fig. \ref{fig00} (f)) is a hyperchaotic one, which begins with an interior crisis after a window composed by a stable period-4 window followed by a quasiperiodic one: $[P_{14}(3.965),P_{15}(3.968)]$. At $P_{16}(3.975)$, through an attractor merging crisis, the size of the hyperchaotic attractor suddenly increases. Such a hyperchatic attractor is presented in Fig. \ref{treispe} at $P_{16}(3.975)$.

At $a=a_{P_{17}}=3.9999$, the hyperchaotic window ends and a narrow chaotic window starts.

\section{Sustaining non-extinction of system dynamics}

\subsection{An unusual chaotic attractor}

Previously it was assumed that, in the absence of prey, the predators become extinct in one generation (see e.g. \cite{lafel}). It is now found numerically that, for some parameter values, $y_n$ might vanish while $x_n$ remains chaotic. Actually, for $a=4$ and $d=3.5$, the system behavior is K--Y chaotic (Fig. \ref{paispe}).
Counterintuitively, after a long chaotic transient (of length $n^*(x_0,y_0)$ of generally thousands or dozens of thousands of iterations, depending on initial conditions $(x_0,y_0)$), the component $y_n$ vanishes, and is trapped by the line $y=0$. The component $x_n$ remains chaotic and fills the line $x\in[0,1]$ (see Figs. \ref{paispe} (a), (b)), with chaotic character becoming stronger (see Fig. \ref{paispe} (d)). The chaotic character is verified here with phase portrait, time series and the '0-1' test. For, $n>n^*$, the attractor is considered a chaotic attractor embedded in the set $[0,1]\times \{0\}$. Denote this chaotic attractor by $\mathcal{A}$.
For $n<n^*$, the shape of $\mathcal{A}$ resembles the shape of the hyperchaotic attractors for $a\in[P_{14},P_{17}]$ (see Fig. \ref {treispe}). Actually, for $n<n^*$, $\mathcal{A}$ is hyperchaotic, while for $n>n^*$, $\mathcal{A}$ becomes K--Y chaotic. The attraction basin of $\mathcal{A}$ is drawn in Fig. \ref{paispe} (c).
Numerically, the points on the basin boundary are disposed along the line $y=-20x+20$. Moreover, it can be proved analytically\footnote{For simplicity, the proof is not presented here.} that the orbits of \eqref{eq1} are located in the domain of $a\ge ax+by$. If $a\ge d$, the last inequality can be easily improved: $ad/4\ge dx+by$. Therefore,  if $a=4,d=3.5$, one gets $20\ge 20x+y$, the line at the attraction basin of $\mathcal{A}$. Moreover, for all attractors with $a>d$, such as $\mathcal{A}$, one obtains a relation for the upper-right boundary: $y=-17.5x+17.5$ (Fig. \ref{paispe} (a)). The boundary line can be found similarly for all attractors with $d<a$.

\subsection{Anticontrol of the chaotic attractor $\mathcal{\mathcal{A}}$ }
Under certain circumstances, for the case of $\mathcal{A}$, when $y_n$ becomes zero, non-chaotic (i.e. regular non-zero or quasiperiodic) or even chaotic behavior can be desirable. Therefore, if one wants to prevent population $y_n$ from extinction, the following anticontrol (chaoticization)\footnote{\emph{Anticontrol}, or \emph{chaoticization}, represents a concept that one can make a given system chaotic or enhance the existing chaos of a chaotic system (see e.g. \cite{chen}).} algorithm can be used: one or both variables, are modified as follows

\begin{equation}\label{antica}
\left\{\begin{array}{l}
x_n=(1-\gamma_1)x_n~~\text {if}~~ \mod (n,\Delta_1)=0\\
y_n=(1-\gamma_2)y_n~~\text {if}~~ \mod (n,\Delta_2)=0,~~n=0,1,2,...
\end{array}\right .
\end{equation}
with $\gamma_{1,2}$ being some small positive real parameters.

Recall that an important aspect of control theory is to ``force'' the system dynamics to any arbitrary targeted regular behavior. Contrarily, here anticontrol is used to enhance chaos of the system, so that it escapes from the ``trap'' $y=0$, without change the dynamics of the variable $x$.

\begin{remark}
Since the perturbations in \eqref{antica} are periodic: $x_n\rightarrow x_n(1-\gamma_1)$ at every $\Delta_1$ steps, and $y_n\rightarrow y_n(1-\gamma_2)$ at every $\Delta_2$ steps, after some $\Delta$ steps, chosen as the least common multiple of $\Delta_1$ and $\Delta_2$, one gets back to the initial stage, so the impulses are $\Delta$-periodic. Therefore, the anticontrolled system represents a discrete system. For example, with $\Delta_1=2$ and $\Delta_2=3$, one has $\Delta=6$. So, after 6 iterations, the procedure is repeating. This means a map $F: (x_0,y_0)\rightarrow(x_6,y_6)$, which determines the dynamics of the system where the $n$-periodic orbits of $F$ give $6n$-periodic orbits of the anticontrolled system. In general, $F : (x_0,y_0)\rightarrow(x_\Delta,y_\Delta)$ represents the discrete version of the anticontrolled system.
\end{remark}

Rather than chaotic dynamics, regular motions can also be obtained (when the algorithm is a chaos control algorithm).

The algorithm is called here the \emph{anticontrol} algorithm, which was introduced by G\"{u}\'{e}mez and Mat\`{\i}as in \cite{gm2} (see also \cite{gm1}), but to control chaotic behavior in discrete and continuous systems. In this paper, the algorithm is used to anticontrol the variable $y_n$ only (see also \cite{in}).

Now consider the simplest case, easily to implement numerically, when only the variable $y_n$ is perturbed at each iteration, i.e. $\Delta_2=1$

\[
y(n)\rightarrow(1-\gamma)y(n),~~~n=0,1,2,...
\]

\noindent The perturbed system has the following initial value problem:
\begin{equation}\label{e1}\left\{\begin{array}{l}
x_{n+1}=ax_n(1-x_n)-b(1-\gamma)x_ny_n,~~~ (x_0,y_0)=(0.5,1.2), \\
y_{n+1}=d(1-\gamma)x_ny_n,
\end{array}\right .
\end{equation}
with $n\in\N$ and nonnegative constants $a,b,d$, and $\gamma$ being a small positive parameter.

The fixed points are

$$
\bar{X}_1^*=(0,0),\quad \bar{X}_2^*=\left(1-\frac{1}{a},0\right) \quad \text{and}\quad \bar{X}_3^*=\left(\frac{1}{d(1-\gamma)},\frac{ad(1-\gamma)-a-d(1-\gamma)}{bd(1-\gamma)^2}\right).
$$
The stability of the fixed points $\bar{X}_{1,2,3}^*$, is analyzed similarly to the case of the fixed points $X_{1,2,3}^*$.

\noindent $\bar{X}_2^*$ is stable if and only if
\[a\in(1,3),\quad d(1-\gamma)<a/{a-1},
\]
and $\bar{X}_3^*$ is stable if and only if
$$\max\bigg \{\frac{3a}{3+a},\frac{a}{a-1}\bigg\}<d(1-\gamma)<\frac{2a}{a-1}.$$
To find the values of $\gamma$ for which the anticontrol algorithm enhances either regular or chaotic motion, the bifurcation diagram, with $\gamma$ as the bifurcation parameter, provides a useful tool.

Next, the dynamics of the anticontrolled system \eqref{e1} is analyzed numerically for $\gamma\in[0.05,0.35]$, one of the most interesting ranges for $\gamma$.

The bifurcation diagram in Fig. \ref{antictrl} shows a reverse scenario compared to the direct bifurcation scenario of uncontrolled system \eqref{eq1} with $a$ as the bifurcation parameter. Let $P_1',P_2',...$, and $SN'$ be some of the most important points in the parameter space. Several reverse cascades of period-doubling bifurcation take place for $\gamma$ between $(P_4',P_7')$. The spectrum of LEs, $\Lambda$, and the $K$ value of the '0-1' test, are used to check the results of the anticontrol algorithm.

Between $(P_1',P_3')=(0.05,0.104)$, the anticontrolled system is hyperchaotic (see the case of $\gamma=0.075$ in Fig. \ref{saispe} (a)). The hyperchaotic window contains the stable window located at $P_2'$ with $\gamma_{P_2'}\approx 0.0958$ (see Fig. \ref{saispe} (b), where for $\gamma=0.0958$ a stable period-20 orbit is generated after a long hyperchaotic transient, which resembles the former hyperchaotic attractor). A chaotic window, containing several narrow periodic windows, begins right after the hyperchaotic one at $P_3'(0.104)$ and ends at $P_5'(0.1128)$. A representative case is obtained for $\gamma=0.107$ (see Fig. \ref{saispe} (c)). Note that at $P_4'(1086)$ an abruptly transition to chaos appears, via (presumably) attractor-merging crisis, without changes in the attractor shapes. For $\gamma\in(P_5'(0.1128),P_8'(0.1771))$, the anticontrolled system has regular orbits (see Fig. \ref{saispe} (d), where a stable period-10 orbit is obtained for $\gamma=0.125$). For $\gamma$ decreasing from $P_8'(0.1771)$, a reverse cascade of period-doubling bifurcation starts and continues, with a period-doubling bifurcation at $P_7'(0.1369)$. Then, at $P_6'(0.1194)$, and so on, it continues till to $P_5'(0.1128)$ where the stable window meets, via a reverse boundary crisis, the chaotic window $(P_3'(0.104), P_5'(0.1128))$. Between $P_8'(0.1771)$ and $NS'$, the system is quasiperiodic. The point $SN'$ is the point where the anticontrolled system undergoes the $N-S$ bifurcation, for $\gamma$ given analytically by Theorem \ref{hopf2}. The large window $(P_8'(0.1771),NS')$ contains several narrow periodic windows. A quasiperiodic orbit, with $\gamma=0.1771$, is presented in Fig. \ref{saispe} (e).
  Note that this attractor resembles the two-tori of two linearly coupled logistic maps (at a larger scale) \cite{alta}. Numerically, the obtained $N-S$ bifurcation parameter value $\gamma=0.2381$ corresponds to the analytical value given by the following $N-S$ theorem for the anticontrolled system \eqref{e1}.

\begin{theorem}\label{hopf2}
For
\[
\gamma=\frac{a d-2 a-d}{(a-1) d},
\]
the fixed point $\bar{X}_3^*$ of the anticontrolled system \eqref{e1} undergoes an $N-S$ bifurcation.

\end{theorem}

\begin{proof}See the proof in Appendix \ref{dd}.
\end{proof}

Comparing to the case of the uncontrolled system, the $N-S$ bifurcation is subcritical and also a single one. Thus, for $\gamma>\gamma_{NS'}$, the system behaves regularly and the orbits are attracted by the fixed point $\bar{X}_3^*$ (see the case of $\gamma=0.25$ in Fig. \ref{saispe} (f) where, with $\gamma=0.25$, $\bar{X}_3^*=(0.381,9.841)$).

Summarizing, the anticontrol algorithm can be used successfully to obtain all kinds of motions, from stable periodic orbits to chaotic orbits.

Results about more general case of the anticontrol algorithm, with $\Delta_1=2$, $\Delta_2=3$ and $\gamma_1=\gamma_2$, are presented in Fig. \ref{saptespe}.

\section* {Conclusion and discussion}
In this paper, by extensive numerical calculations and some analysis, extremely rich dynamics of a prey-predator system are revealed. The system presents regular motions (stable periodic orbits and quasiperiodic orbits) and chaotic and hyperchaotic dynamics. Some chaotic attractors are in the sense of Kaplan-Yorke (with positive sum of LEs).

Moreover, with $a=4$, $d=3.5$ and $b=0.2$, a particular chaotic attractor, $\mathcal{A}$, is found for which the $y_n$ component vanishes after a relatively large number of iterations. If this kind of extinction should be prevented, an anticontrol algorithm, used before to control chaos, can be utilized to either enhance chaos (or hyperchaos) existing in the system, or to a reach some regular orbit. To verify the accuracy of the numerical results, beside the analytically studied on the Neimark-Sacker bifurcation, several numerical tools have been applied: time series, phase portraits, power spectral density and the '0-1' test.

The '0-1' test, utilized to unveil the chaotic of the attractor $\mathcal{A}$, indicates, for the first time in the literature according to the authors' knowledge, an interesting behavior of $K$ at the $N-S$ bifurcation point (green circles at points $O$ and $O'$ on Fig. \ref{fig1} and Fig. \ref{antictrl}, respectively). Thus, usually at this $N-S$ point, $K$ should be (or very close to) zero since right before the bifurcation (supercritical bifurcation) or after the bifurcation (subcritical bifurcation), the system orbit is regular with $K$=0. Therefore, this phenomenon deserves further analytical and numerical investigations in the future.

Another interesting discovery, revealed by LEs and the '0-1' test, is that for the uncontrolled system, with $a\in(3.2,3.25)$, the largest LE is nonpositive and $K$ takes intermediate values between 0 and 1. This result indicates the presence of SNAs, which usually appear in systems externally driven by two incommensurate frequencies (see e.g. \cite{prasa}, and references therein, where it is shown how the '0-1' test can be used to detect SNAs). These attractors, ``intermediate'' between strange chaotic attractors and nonchaotic regular dynamics, are geometrically strange because they can be properly described by some fractal dimensions. It is well known that mathematically proving the existence of such attractors is a nontrivial task. Even for a relatively large range of the parameter $\gamma$, the anticontrolled system presents $K$ values intermediate between 0 and 1, because the corresponding LEs are positive (see Fig. \ref{antictrl}) hence there are no SNAs. This suggests that the anticontrol algorithm would destroy the SNAs of the uncontrolled system, which could be useful in practical applications when these kind of attractors are not desirable.

SNAs are generic in quasiperiodically-driven nonlinear systems, which have the largest LE being zero or negative. Therefore, trajectories (i.e. orbits) do not show exponential sensitivity to initial conditions and they are not chaotic (see also \cite{prasa,laka} and related references). Their geometric structure is fractal. Usually, SNAs connect, for a relatively large parameter range, quasiperiodic attractors to chaotic attractors and can be detected using the sign of the largest LE and the $K$ values.

Compared with the classical studies on discrete systems with SNAs (such as the quasiperiodicically forced logistic map), where the existence of these attractors is defined for relatively large parameter ranges, in the non-quasiperiodically driven prey-predator system \eqref{e1}, this phenomenon appears only on very short parameter ranges, or possibly just at some isolated points. Therefore, without a deeper analytical or numerical study, one can only presumably conclude that the system admits SNAs. However, by considering the statements in \cite{sna}\footnote{``We conjecture that, in general, continuous time systems (``flows'') which are not externally driven at two incommensurate frequencies should not be expected to have SNAs except possibly on a set of measure zero in the parameter space.''}, a further study is in order.

 For $a_{SNA}=3.2096852$, within a very small parameter range, the largest LE is negative and $K\thickapprox 0.5$. Figs. \ref{ultim} (a), (b) present the variations of LEs and $K$ within a small interval around $a_{SNA}$. Figs. \ref{ultim} (c), (d) show the plots of $p$ versus $q$ and the mean square displacement $M$ as a function of $n$, which are typical for SNAs (compare with Fig. \ref{apendice} (c), where the SNA of the famous GOPY model introduced by Grebogi et al. in \cite{sna} is considered). There are several other small parameter ranges, where SNAs could appear.

 Note that, with $a=3.3$ (Fig. \ref{noua} (f)), the shape of the attractor in the phase space resembles some SNAs (see e.g. \cite{laka}). Therefore, a subtle investigation for $a$, slightly different from $a=3.3$, could reveal SNAs.

Beside the bifurcations at points $NS$, $NS_1$ and $NS_2$, determined both analytically and numerically (see Subsection \ref{maimulte} (i) and Fig. \ref{fig00} (c)), other presumably $N-S$ bifurcations have been found numerically. At points $SN_3$ and $SN_4$, a subcritical and supercritical $N-S$ bifurcations, respectively, have been found in the window $a\in(3.2,3.215)$, where multilayered branches of quasiperiodic orbits can be seen (Fig. \ref{ultima_sn}). Note that in the middle and at the end of this window, direct and reverse period-doubling bifurcations take place. Also, beside the attractors coexistence shown in Fig. \ref{fig00} (d), three other $N-S$ supercritical bifurcations can be seen (points $NS_5$, $NS_6$ and $NS_7$). If one denotes by $F$ the map associated with system \eqref{e1}, then these dynamics could be interpreted as $N-S$ bifurcations of some fixed points of $F^n$, which generates a family of invariant circles.

As is well known, the $N-S$ bifurcation induces a route to chaos. Thus, via dynamic transition from the fixed point $X_3^*$ to quasiperiodic windows, and to periodic windows of mode-locked orbits, occurring in between, the system becomes chaotic through quasiperiodicity (see e.g. Figs. \ref{noua} (b), (d) and Figs. \ref{unspe} (ii), (iii)). This routh to chaos could also be explained by the birth of new frequencies in the PSD or via the crisis scenario (crisis route).

Another interesting, new and active topic research is, as well known, spatial effects are important for predator-prey systems. Thus, in terrestrial or macroscopic marine predator-prey systems the predators typically disperse more rapidly than their prey \cite{sup3,sup4,sup5} (see also \cite{sup1,sup2}). This fact seems to be strongly related with the vanishing of the predator $y$ in the considered system. Therefore, a future study on how the proposed anticontrol algorithm could be used to other systems, could be useful and important.

\renewcommand{\thesubsection}{\Alph{subsection}}

\counterwithin{equation}{subsection}


\section*{Appendix}

\addtocounter{subsection}{-2}
\subsection{Stability of $X_{1,23}^*$}\label{aa}

Denote by $e$ the eigenvalues. As is well known, a fixed point $X^*$ of a discrete system is stable if and only if its eigenvalues, $e$, satisfy the condition of $|e|<1$.

(i) For $X_1^*$, eigenvalues are $e_1=0$ and $e_2(a)=a$. Therefore, $X_1^*$ is stable for $a<1$.

(ii) For $X_2^*$, eigenvalues are $e_1(a)=2-a$ and $e_2(a,d)=(a-1)d/a$. Therefore, for $a\in(1,3)$ and $d<a/(a-1)$, the fixed point $X_2^*$ is stable.

(iii) For $X_3^*$ with $a>1$, the graphs considered next are obtained by representing $d$ as a function of $a$: $d=d(a)$ (Fig. \ref{fig1} (b)).
Denote $\Delta(a,d)=(a/d+2)^2-4a$ (see the graph \textcircled{\raisebox{-0.9pt}{4}} in Fig. \ref{fig1} (b)).
If $\Delta(a,b)\geq0$ (regions \textcircled{\raisebox{-0.9pt}{3}}, \textcircled{\raisebox{-0.9pt}{5}} and \textcircled{\raisebox{-0.9pt}{8}}, including the blue region), the eigenvalues are real: $e_{1,2}^r(a,d)=\Big(1-\frac{a}{2d}\Big)\pm \frac{1}{2}\sqrt{\Delta(a,d)}$.
If $\Delta(a,d)< 0$ (regions \textcircled{\raisebox{-0.9pt}{2}}, \textcircled{\raisebox{-0.9pt}{1}} and \textcircled{\raisebox{-0.9pt}{6}}), the eigenvalues $e_{1,2}^c$ are complex conjugated, $e_1^c=\overline{e_2^c}$: $e_{1,2}^c(a,d)=\Big(1-\frac{a}{2d}\Big)\pm \imath\frac{1}{2}\sqrt{-\Delta(a,d)}$.
Modulus of the complex (and also of the real) eigenvalue is
\begin{equation}\label{modul}
A(a,d)=\sqrt{a-\frac{2a}{d}}.
\end{equation}
Finally, after some simple calculations, omitted here, the stability parametric domain is found (regions \textcircled{\raisebox{-0.9pt}{2}}, \textcircled{\raisebox{-0.9pt}{3}} and \textcircled{\raisebox{-0.9pt}{4}}, with boarders \textcircled{\raisebox{-0.9pt}{1}} and \textcircled{\raisebox{-0.9pt}{5}} respectively):

$$\max\bigg \{\frac{3a}{3+a},\frac{a}{a-1}\bigg\}<d<\frac{2a}{a-1}.$$
In the remaining regions, \textcircled{\raisebox{-0.9pt}{6}} and \textcircled{\raisebox{-0.9pt}{8}}, fixed points are unstable.

\numberwithin{equation}{subsection}

\subsection{Proof of Theorem \ref{ho1}} \label{bb}

\noindent (i) The bifurcation condition, $|e_{1,2}^c(a,d)|=1$, i.e. $A(a,d)=1$ (see \eqref{modul}), leads to $d-2a/(a-1)=0$;

\noindent (ii) The derivative of $e_{1,2}^c$ with respect to $d$, determined along AB, is
\[\frac{\partial}{\partial d}\Big\rvert e_{1,2}^c(a,d)\Big\rvert_{d=2a/(a-1)}=\frac{a}{d^2\sqrt{a-2a/d}}\Bigg\rvert_{d=2a/(a-1)}=\frac{(a-1)^2}{4a}:=k>0~ \text{for~ all}~ a>1;
\]
\noindent (iii) On the curve AB, the equation $(e_{1,2}^c(a,d))^j=1$, for $d=2a/(a-1)$, has the following solutions: for $j=1$, $a=1$; for $j=2$, $a\in\{1,9\}$; for $j=3$, $a\in\{1,7\}$, and for $j=4$, $a\in\{1,5,9\}$. Therefore, on the considered domain $a\in(1,4]$ and $d=2a/(a-1)$ (see Fig. \ref{fig1} (a)), the condition (iii) is satisfied.

\subsection{Neimark-Sacker (N$-$S) bifurcation of the anticontrolled system}\label{dd}

\noindent (i) $A(\gamma)=1$ if and only if
\begin{equation*}\label{e2b}
\gamma=\frac{a d-2 a-d}{(a-1) d}.
\end{equation*}
The graph of $\gamma$ in the space $(a,d,\gamma)$, $\mathcal{S}$, is the $N-S$ surface (Fig. \ref{fig1} (d)).

Note that the intersection with the plane $\gamma=0$ (Fig. \ref{fig1} (d)) represents the $N-S$ curve for the uncontrolled system \eqref{eq1}, at the fixed point $X_3^*$.

\noindent (ii) Next, one has
$$
\frac{\partial}{\partial \gamma}\Big\rvert e_{1,2}^c(\gamma)\Big\rvert_{\gamma}=-\frac{(a-1)^2 d}{2 a}:=k<0.
$$
\noindent (iii)
To ensure $\gamma\in(0,1)$, suppose $d>2a/(a-1)$. For $a\in(1,9)$, the eigenvalues $e_{1,2}^c(\gamma)$ are: $e_{1,2}^c(\gamma)=\frac{5-a}{4}\pm\imath\frac{1}{4}\sqrt{(9-a)(a-1)}$.
Note that the image of the fixed point $e_{1}^c(\gamma)=\frac{5-a}{4}+\imath\frac{1}{4} \sqrt{(9-a)(a-1)}$ is in the half upper complex-plane, hence (compare with \eqref{alfa})
$$
\arg(e_{1,2}^c(\gamma))=\arccos\frac{5-a}{4}.
$$
Since $\arccos\frac{5-a}{4}$ is increasing from $0$ to $\pi$ for $a\in(1,9)$, the equation $(e_{1,2}(\gamma))^{j}=1$ has only the following solutions: $j=3$, $a=7$, $j=4$, $a=5$.

\subsection{The '0-1' test}\label{ee}

The '0-1'  test, proposed in \cite{01}, is designed to distinguish chaotic behavior from regular behavior in deterministic systems.
Consider a discrete or continuous-time dynamical system and a one-dimensional observable data set $\phi(j)$, $j=1,2,...,N$, of the underlying system, constructed from time series. The 0-1 test has a theorem \cite{012}, which states that a nonchaotic motion is bounded, while a chaotic dynamic behaves like a Brownian motion.

\noindent 1) First, compute the translation variables (for some $c\in(0,\pi)$, \cite{01})
\[
p(n)=\sum_{j=1}^n\phi(j)cos(jc),~~~ q(n)=\sum_{j=1}^n\phi(j)sin(jc),
\]
for $n=1,2,...,N$.

\noindent 2) To determine the growths of $p$ and $q$, the mean-square displacement is determined:
\[
M(n)=\lim_{N\rightarrow \infty}\frac{1}{N}\sum_{j=1}^N[p(j+n)-p(j)]^2+[q(j+n)-q(j)]^2.
\]
where $n\ll N$ (in practice, $n=N/10$ gives good results).

\noindent 3) The asymptotic growth rate is defined as
\[
K=\lim_{n\rightarrow \infty}\log M(n)/\log n.
\]
Because of occurrence of resonances for isolated values of $c$ (where $K$ is larger), the median of the computed values of $K$ is used, since the median is robust against outliers associated with resonances \cite{01}.
If the underlying dynamics is regular (i.e. periodic or quasiperiodic) then $K = 0$; if the underlying dynamics is chaotic then $K = 1$.
Improved variants can be found in \cite{01,014}.

In Fig. \ref{apendice}, the GOPY model are presented. Fig. \ref{apendice} (a) represents the plot of $q$ versus $p$ and in Fig. \ref{apendice} (b) the mean-square displacement $M$ as a function of $n$. Figs. (i) represent the regular orbit of the logistic map $x_{n+1}=rx_n(1-x_n)$ for $r=3.55$; Figs. (ii) present the chaotic orbit of the logistic map for $r=4$, while Figs. (iii) the SNA of the GOPY map $x_{n+1}=2a\tanh(x_n)\cos(2\pi\theta_n)$, $\theta_{n+1}=\theta_n+\omega$, with $a=1.5$, and $\omega=(\sqrt{5}-1)/2$ \cite{sna}.

\textbf{Acknowledgement} This study was partially funded by Russian Science Foundation project 19-41-02002.

Conflict of Interest: The authors declare that they have no conflict of interest.

\newpage{\pagestyle{empty}\cleardoublepage}

\begin{figure}
\begin{center}
\includegraphics[scale=1]{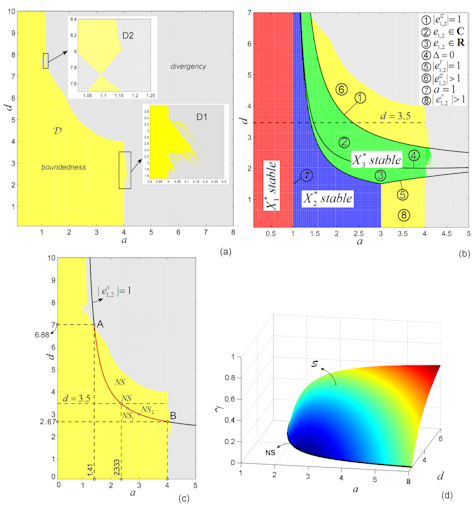}
\caption{(a) Boundedness parametric domain of the system \eqref{eq1}. Yellow presents the boundedness domain and grey the divergency domain. (b) Stability domains of the fixed points $X_{1,2,3}^*$. The stability domain of the fixed point $X_1^*$ is plotted in red, of the fixed point $X_2^*$ in blue, while the fixed point $X_3^*$ in green. Curves \textcircled{\raisebox{-0.9pt}{2}} and \textcircled{\raisebox{-0.9pt}{3}} represent the stability domains of fixed point $X_3^*$ (region \textcircled{\raisebox{-0.9pt}{2}} represents complex eigenvalues, \textcircled{\raisebox{-0.9pt}{3}} real eigenvalues). Regions \textcircled{\raisebox{-0.9pt}{6}} and \textcircled{\raisebox{-0.9pt}{8}} represent the instability domains. (c) The Neimark-Sacker curve $AB$ with $A(1.41,6.88)$ and $B(4,2.67)$, and the $N-S$ bifurcation points $NS$ $a=2.3333$ and $d=3.5$, $NS_1$, with $a=2.507...$, and $b=3.327...$, and $HS_2$ with $a=3$ and $d=3$. (d) The Neimark-Sacker surface of the anticontrolled system \eqref{e1}}
\label{fig1}
\end{center}
\end{figure}

\begin{figure}
\begin{center}
\includegraphics[scale=1]{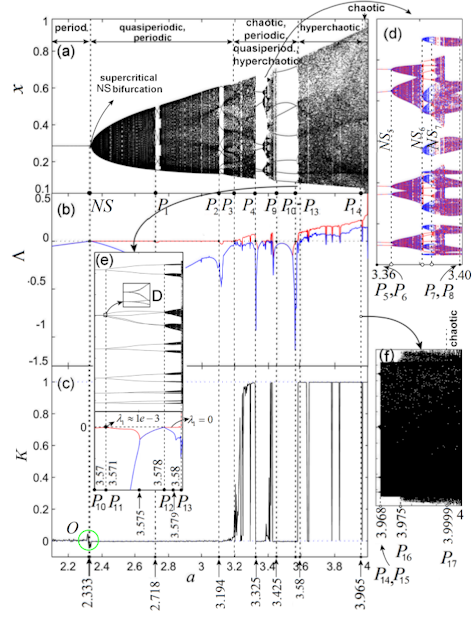}
\caption{(a) Bifurcation diagram of the component $x_n$. (b) Lyapunov spectrum $\Lambda$. (c) Asymptotic growth rate $K$ of the 0-1 test. Detailed zone (d) presents a multistability window in the bifurcation diagram, while (e) presents a zoomed detail of window $(P_6,P_7)$. In detail, (f) presents a zoomed area around the point $P_8$ }
\label{fig00}
\end{center}
\end{figure}

\begin{figure}
\begin{center}
\includegraphics[scale=1]{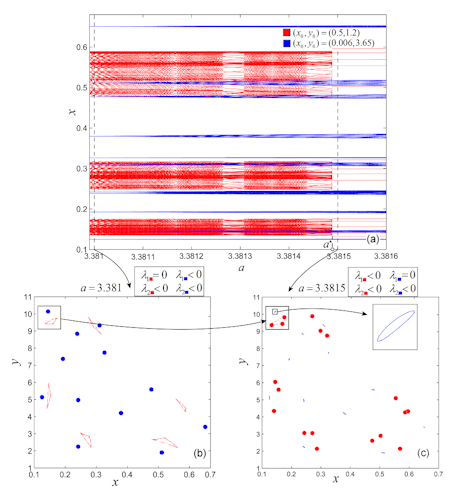}
\caption{Coexisting stable orbits and quasiperiodic orbits for $a\in(3.380,3.3816)$. (a) Bifurcation diagram showing two coexisting attractors generated from two different attraction basins (red and blue, respectively). (b) Phase overplot of the orbits starting from $(x_0,y_0)=(0.5,1.2)$ (red plot) and $(x_0,y_0)=(0.006,3.65)$ (blue plot), for $a=3.381$. The stable period-12 orbit is drawn by filled blue circles, while the quasiperiodic orbit is plotted in red. (c)  Phase overplot of the stable period-18 orbit starting from $(0.5,1.2)$ (red plot) and the 12 islands-like of the quasiperiodic orbit starting from $(0.006,3.65)$ (blue plot), for $a=3.3815$. Rectangular zones underline the transformation from quasiperiodic motion to stable periodic orbit, and vice versa }
\label{figg}
\end{center}
\end{figure}

\begin{figure}
\begin{center}
\includegraphics[scale=1]{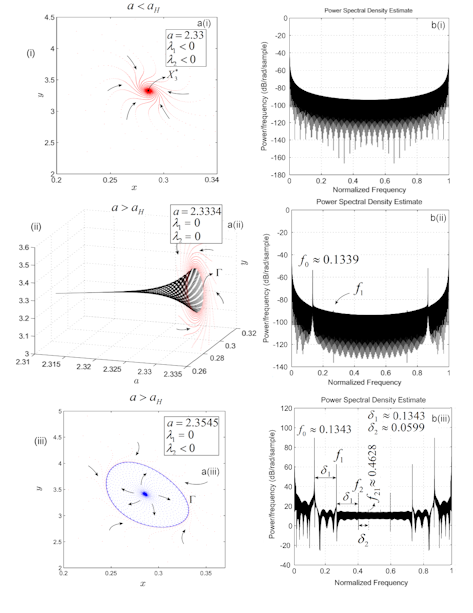}
\caption{Neimark-Sacker bifurcation: Figs. (a) present the phase portraits and Figs. (b), PSD.
(i) $a=2.33$ ($a<a_H$); The phase portrait in Fig. \ref{sapte} a(i) shows that $X_3^*$ is still attractive. In Fig. \ref{sapte} b(i), the PSD reveals that for $a<a_H$ the orbits have no frequencies. (ii)  $a=2.3334$ (slightly after $a_H$): The three-dimensional view shows that orbits from outside of the quasiperiodic curve $\Gamma$ are attracted by the invariant circle born after N$-$S bifurcation. PSD reveals the birth of the frequency $f_0$ and the place of the future harmonic. (iii) $a=2.3545$. Once $a$ is incremented, next harmonics of $f_0$, $f_1$ and $f_2$ are born equidistantly with offset $\delta_1$, and the second main frequency $f_{21}$ appears with distance $\delta_2$ from $f_2$}
\label{sapte}
\end{center}
\end{figure}


\begin{figure}
\begin{center}
\includegraphics[scale=1]{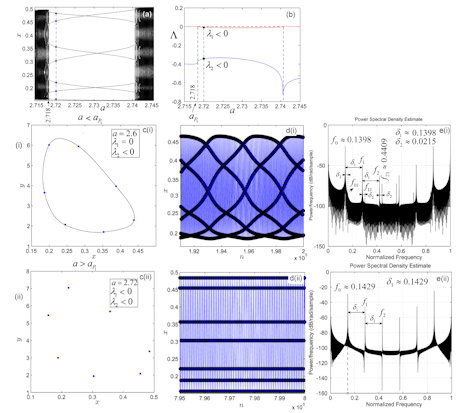}
\caption{Bifurcation point $P_1(2.718)$. (a) Bifurcation diagram for $a\in[2.715,2.745]$, containing the point $P_1$. (b) Lyapunov spectrum for $a\in[2.715,2.745]$.
(i) $a=2.6<a_{P_1}$. (ii) $a=2.72>a_{P_1}$. (c) Phase portraits; (d) Time series; (e) PSD. For $a=2.6$, at about the middle interval $(a_{H},a_{P_1})$, the motion is a quasiperiodic orbit with basic frequency $f_{0}$ and its modulation frequencies (harmonics), multiple of $f_0$, satisfying $f_{k}=f_0+k\delta_1=(k+1)f_0$, $k=1,2$, with offset $\delta_1>0$. Also, the second frequency has harmonics $f_{k1}=f_{k}+\delta_2$, $k=0,1$. The quasiperiodic oscillations are non-resonant and contains 7 unstable fixed points (circled points). For $a=2.72$, the motion became a stable periodic orbit and only frequencies $f_{0,1,2}$ are visible. The oscillations are resonant (mode-locking). The 7 fixed points are stable now and compose a stable period-7 orbit. Point position is slightly different because of the parameter $a$ slight increases}
\label{opt}
\end{center}
\end{figure}

\begin{figure}
\begin{center}
\includegraphics[scale=1]{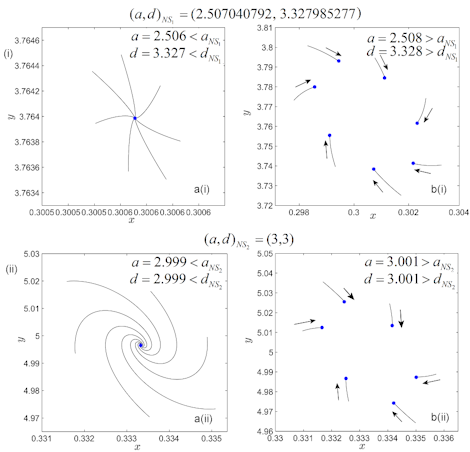}
\caption{Phase portraits of two near points, $NS_1$ and $NS_2$, along the Neimark-Sacker curve $AB$. (i) $(a,d)$ taken near $(a,d)_{NS_1}$. (ii) $(a,d)$ taken near $(a,d)_{NS_2}$. (a) Values of $a$ and $d$ are slightly smaller than $a_{NS_i}$, $d_{NS_i}$, $i=1,2$, respectively. (b) (a) Values of $a$ and $d$ are slightly bigger than $a_{NS_i}$, and $d_{NS_i}$, $i=1,2$, respectively. }
\label{cica_patru}
\end{center}
\end{figure}


\begin{figure}
\begin{center}
\includegraphics[scale=1]{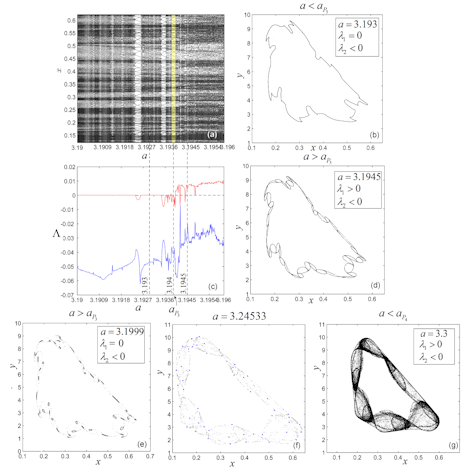}
\caption{Point $P_3(3.194)$. The small stable periodic window (yellow plot), centered at about $P_3$. (a) Bifurcation diagram of $x_n$ variable for $a\in[3.19,3.196]$. (b) Phase plot of a quasiperiodic orbit for $a=3.193<a_{P_3}$. (c) Lyapunov spectrum $\Lambda$ for $a\in[3.19,3.196]$; (d) Phase plot of a chaotic orbit for $a=3.1945>a_{P_3}$. (e) Phase plot of a quasiperiodic attractor for $a=3.1999$, composed by 45 invariant circles. (f) For $a=3.3<a_{P_4}$, because $\sum\lambda_i=0.064-0.015>0$, the attractor is K--Y chaotic }
\label{noua}
\end{center}
\end{figure}

\begin{figure}
\begin{center}
\includegraphics[scale=1]{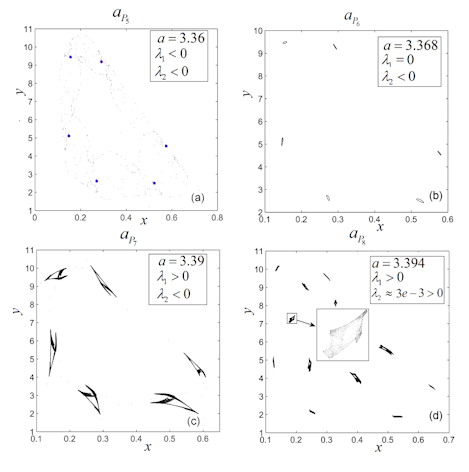}
\caption{Phase portraits for points $P_{5-8}$. (a) Stable orbit for $a=3.36$. (b) Quasi-periodic orbit for $a=3.368$. (c) Chaotic attractor for $a=3.39$. (d) Hyperchaotic attractor for $a=3.394$. Because in this case $\lambda_2=3e-3$, the attractor is weak hyperchaotic. While the chaotic orbit (Fig. \ref{zece} (c)) consists of 6 islands-like sets, the hyperchaotic attractor (Fig. \ref{zece} (d)) is composed of 12 islands-like sets, each of them with different shape (see the enlarge detail). For $a=3.36$, the attractor are reached after long chaotic transient, which presents the shape of the K--Y chaotic attractor at $a=3.3$ (Fig. \ref{noua} (f))}
\label{zece}
\end{center}
\end{figure}

\begin{figure}
\begin{center}
\includegraphics[scale=1]{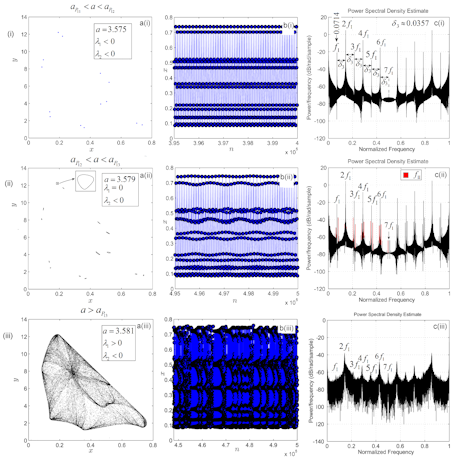}
\caption{Attractors in the narrow window $[P_{10},P_{13}]$. (a) Phase portraits. (b) Time series. (c) Power spectral density. (i) $a\in(a_{P_{11}},a_{P_{12}})$. For $a=3.575$ the orbit is a stable period-14 orbit. Now, the previous two main frequencies $f_0$ and $f_{01}$, where their harmonics have unified, generate the main frequency $f_I\approx0.07143$ and its harmonics $kf_I$, for $k=2,...7$ with offset $\delta_3\approx0.0357$. (ii) $a\in(a_{P_{12}},a_{P_{13}})$. For $a=3.579$ the orbit is quasiperiodic (see the zoomed detail); two symmetric series of new peaks appear in PSD (red plot), revealing the birth of the second main frequency $f_{II}$. However, new frequencies tend to appear, indicating the proximity of chaos. (iii) $a>a_{P_{13}}$. For $a=3.581$, the attractor is K==Y chaotic: $\sum\lambda_i=0.065-0.060>0$. In the broadband of PSD, former frequencies are still visible }
\label{unspe}
\end{center}
\end{figure}

\begin{figure}
\begin{center}
\includegraphics[scale=1]{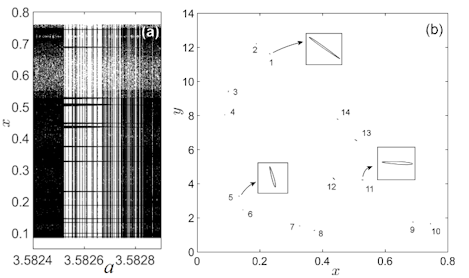}
\caption{(a) Bifurcation diagram for $a\in[3.5824,3.5829]$. (b) Phase portrait of a quasiperiodic orbit for $a$ close to $3.5826$. The enlarged views unveil that there are 14 islands-like sets}
\label{doispe}
\end{center}
\end{figure}

\begin{figure}
\begin{center}
\includegraphics[scale=0.75]{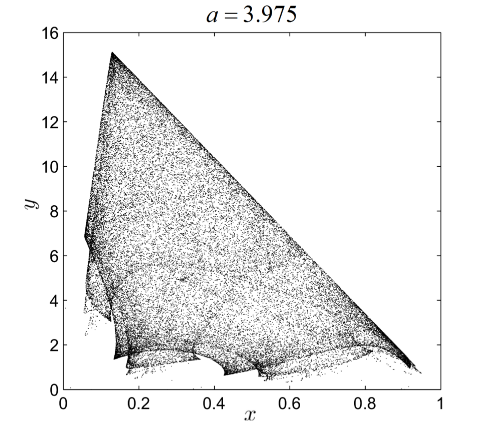}
\caption{Hyperchaotic attractor in the window $P_{14}(3.965),P_{17}(3.9999)]$, for $a=a_{P_{16}}(3.975)$ }
\label{treispe}
\end{center}
\end{figure}

\begin{figure}
\begin{center}
\includegraphics[scale=1]{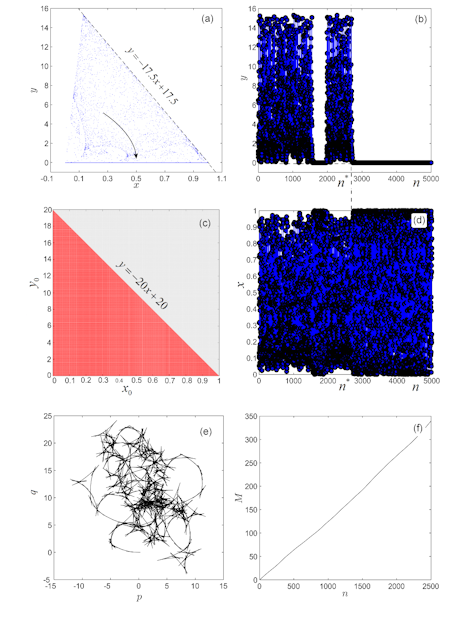}
\caption{The K==Y chaotic attractor $\mathcal{A}$, for $a=4$: $\sum\lambda_i=0.693-0.147>0$. For $n>n^*(x_0,y_0)$, with $n^*(x_0,y_0)$ being some positive integer, the attractor is trapped by the horizontal line $[0,1]$. (a) Phase portrait (the arrow shows the attractor tendency). (b) Time series of the component $y_n$ (for clarity, only the first 5000 iterations are plotted). (c) Attraction basin (red plot) of the chaotic transient attractor. The numerically found frontier of the attraction basin coincides with the line $y=-20x+20$. (d) Time series of the component $x_n$. Dotted vertical line in Figs. (d) and (b) indicates the moment when the algorithm applied (at $n=n^*$); (e) Dynamics of the translation components $(p,q)$ in terms of 0-1 test; (f) Mean-square displacement $M$ as a function of $n$}
\label{paispe}
\end{center}
\end{figure}

\begin{figure}
\begin{center}
\includegraphics[scale=.8]{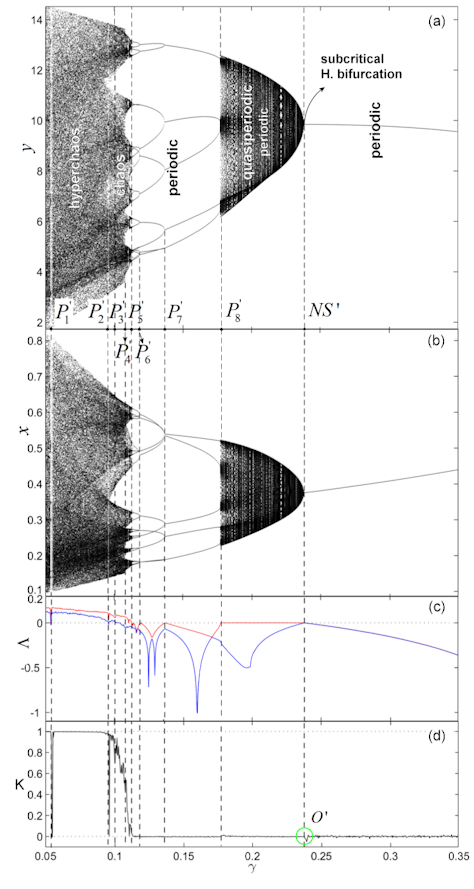}
\caption{Variation of the anticontrolled system \eqref{e1} with $d=3.5$ and $a=4$, versus $\gamma$. (a) and (b) components $x_n$ and $y_n$ respectively. (c) Lyapunov spectrum $\Lambda$. (d) Variation of the asymptotic growth rate $K$ fromthe `0-1' test}
\label{antictrl}
\end{center}
\end{figure}

\begin{figure}
\begin{center}
\includegraphics[scale=1]{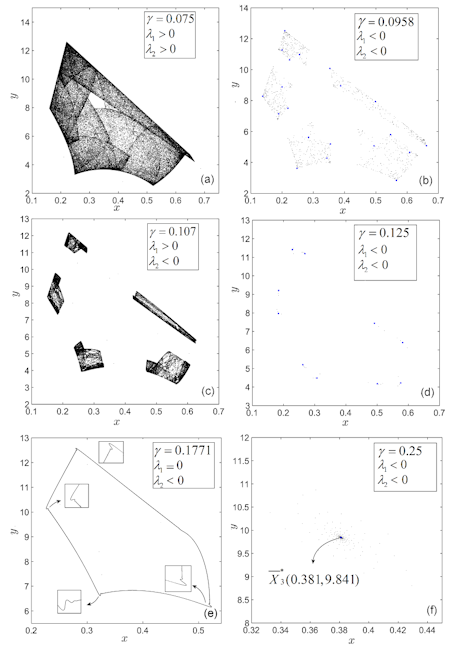}
\caption{Anticontrolled attractors, phase poirtraits. (a) Anticontrolled hyperchaotic attractor for $\gamma=0.075$; both LEs are positive, $\lambda_{1,2}>0$. (b) Anticontrolled period-20 stable periodic orbit for $\gamma=0.0958$; $\lambda_1<0$, $\lambda_2<0$. The (hyper)chaotic transients to the orbit points (blue dots) are reminiscences for the former hyperchaotic attractor. (c) Anticontrolled chaotic attractor for $\gamma=0.107$; $\lambda_1>0$, $\lambda_2<0$. (d) Anticontrolled stable period-10 orbit for $\gamma=0.125$; $\lambda_1<0$, $\lambda_2<0$. (e) Anticontrolled quasiperiodic attractor even after the beginning of the quasiperiodic window $(P_8',NS')$. (f) Anticontrolled stable fixed point for $\gamma=0.25>\gamma_{NS'}$ }
\label{saispe}
\end{center}
\end{figure}


\begin{figure}
\begin{center}
\includegraphics[scale=.8]{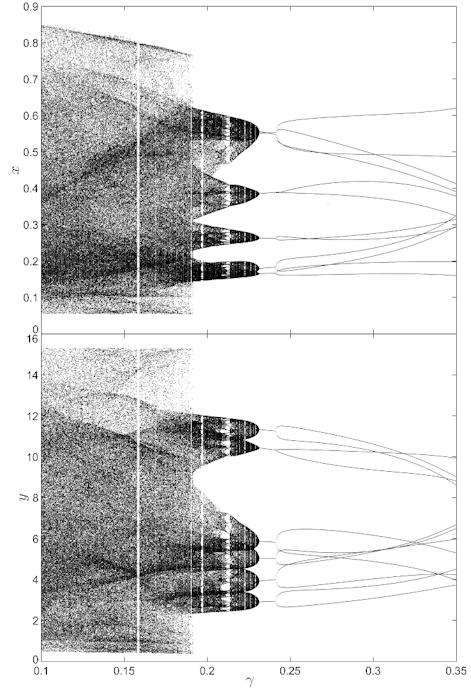}
\caption{Anticontrol of system \eqref{e1} for $d=3.5$ and $a=4$, with perturbations of $x_n$ variable at every $\Delta_1=2$ steps and of $y_n$ variable at every $\Delta_2=3$ steps. (a) Bifurcation diagram of the $x_n$ variable. (b) Bifurcation diagram of the second variable $y_n$ }
\label{saptespe}
\end{center}
\end{figure}

\begin{figure}
\begin{center}
\includegraphics[scale=1]{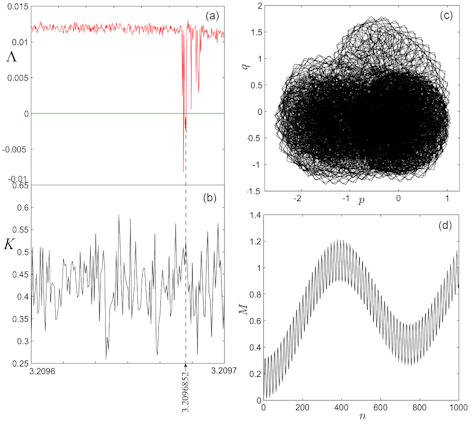}
\caption{SNA of the prey-predator system \eqref{eq1} for $a=3.2096852$. (a), (b) LEs and $K$ value within a small neighborhood of $a=3.2096852$, respectively. (c) Plot of $q$ versus $p$. (b) Mean-square displacement $M$ as a function of $n$}
\label{ultim}
\end{center}
\end{figure}

\begin{figure}
\begin{center}
\includegraphics[scale=.8]{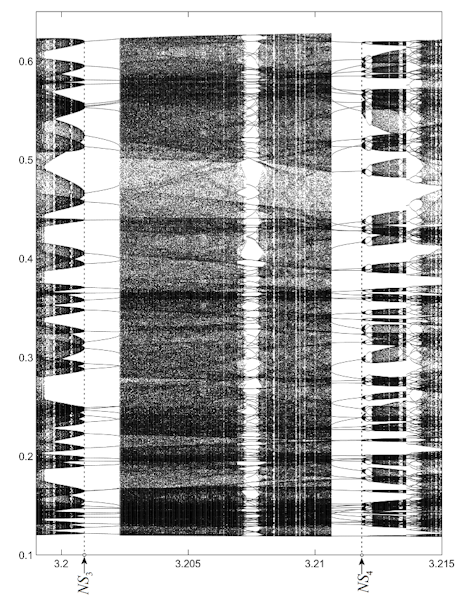}
\caption{Bifurcation diagram for $a\in(3.199,3.215)$. Within this parameter $a$ range, the system undergoes two presumably $N-S$ bifurcations at points $NS_3$ and $NS_4$ }
\label{ultima_sn}
\end{center}
\end{figure}

\begin{figure}
\begin{center}
\includegraphics[scale=1]{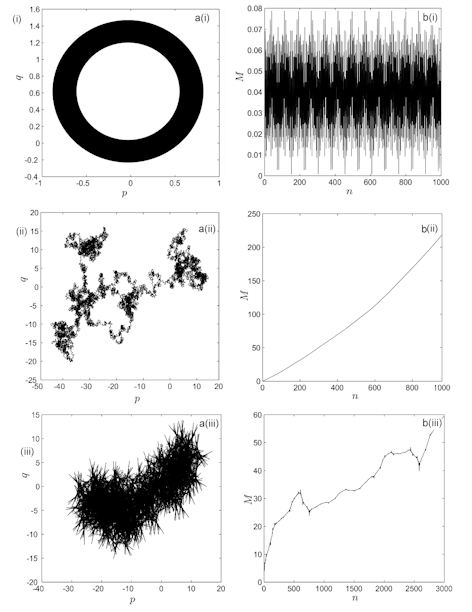}
\caption{The `0-1' test. (a) Plot of $q$ versus $p$. (b) Mean-square displacement $M$ as a function of $n$. (i) Regular dynamics of the logistic map $x_{n+1}=rx_n(1-x_n)$ for $r=3.55$. (ii) Chaotic dynamics of the logistic map, for $r=4$. (iii) SNA of the GOPY model $x_{n+1}=2a\tanh(x_n)\cos(2\pi\theta_n)$, $\theta_{n+1}=\theta_n+\omega$, for $a=1.5$, and $\omega=(\sqrt{5}-1)/2$ \cite{sna}}
\label{apendice}
\end{center}
\end{figure}

\end{document}